\documentclass[11pt]{amsart}
\usepackage{mathpple}
\usepackage{mathrsfs}
\usepackage{amsmath, amscd, amsthm,amssymb, thmtools, amsfonts, verbatim}
\usepackage[mathcal]{eucal}
\usepackage{hyperref}
\usepackage{relsize}
\usepackage{tikz-cd}
\usepackage[cmtip, arrow, all]{xy}
\usepackage{pb-diagram ,pb-xy}
\usepackage{graphicx}
\usepackage{subcaption}
\usepackage{epstopdf}
\usepackage{color}
\usepackage{ifthen}

\renewcommand{\epsilon}{\varepsilon}

\newcommand{\C}{\mathbb C}

\newcommand{\Z}{\mathbb Z}

\newcommand{\op}{\operatorname}

\newcommand{\mbb}{\mathbb}

\newcommand{\ip}[1][\cdot,\cdot]{\left\langle #1 \right\rangle}

\newcommand{\R}{\mbb R}

\renewcommand{\Im}{\op{Im}}

\newcommand{\defretract}[5]{
\begin{tikzcd}[ampersand replacement=\&]
#1\arrow[r,shift left = .5 ex,"#3"] \&  #2\arrow[l, shift left = .5 ex,"#4"]\arrow[loop right, distance = 4em,start anchor = {[yshift = 1ex]east},end anchor = {[yshift=-1ex]east}]{}{#5}
\end{tikzcd}
}

\DeclareMathOperator{\Aut}{Aut} 
 
\DeclareMathOperator{\Sym}{Sym} \DeclareMathOperator{\Hom}{Hom}

\DeclareMathOperator{\coker}{coker}

\def\sO{\mathscr O}
\def\sS{\mathscr S}
\def\sV{\mathscr V}

\def\fJ(E){\mathfrak E}
\def\fJ{\mathfrak J}

\def\fg{\mathfrak g}

\declaretheoremstyle[
spaceabove=7pt, spacebelow=7pt,
headfont=\normalfont\bfseries,
notefont=\mdseries, notebraces={(}{)},
bodyfont=\itshape,
postheadspace=5pt,
headpunct = .
]{thm}

\declaretheoremstyle[
spaceabove=7pt, spacebelow=7pt,
headfont=\normalfont\bfseries,
notefont=\mdseries, notebraces={(}{)},
bodyfont=\itshape,
postheadspace=10pt,
headpunct = .
]{def}

\declaretheoremstyle[
spaceabove=4pt, spacebelow=7pt,
headfont=\itshape,
postheadspace=5pt,
headpunct = :,
postheadspace = 3pt, qed = $\lozenge$
]{rem}

\declaretheorem[numbered = yes, parent = section, style = thm]{theorem}

\declaretheorem[numbered = no, style = thm, name = Corollary A]{corA}

\declaretheorem[sibling = theorem, style = thm, name = Theorem/Definition]{thm-def}
\declaretheorem[sibling = theorem, style = thm]{proposition}
\declaretheorem[sibling = theorem, style = thm]{lemma}

\numberwithin{equation}{section}

\declaretheorem[sibling = theorem, style = def]{definition}

\declaretheorem[sibling = theorem, style = rem]{remark}

\newcommand{\cinfty}{C^{\infty}}
\newcommand\Obcl[1][~]{\ifthenelse{ \equal{#1}{~}} {
  \operatorname{Obs}^{cl}
}{
  \operatorname{Obs}^{cl}(#1)
}}
\newcommand\Obq[1][~]{\ifthenelse{ \equal{#1}{~}} {
  \operatorname{Obs}^{q}
}{
  \operatorname{Obs}^{q}(#1)
}}

\newcommand\alg{\sO(V)}

\newcommand{\shObq}{\mathscr{O}bs^q}
\newcommand{\shObcl}{\mathscr{O}bs^{cl}}
\newcommand{\Pflm}{\det{( D_{\lambda}^\mu)^+}}
\DeclareMathOperator{\pf}{pf}
\DeclareMathOperator{\rk}{rk}
\DeclareMathOperator{\Det}{Det}

\title[The BV Formalism and the Determinant Line]{The Batalin-Vilkovisky Formalism and the Determinant Line Bundle}
\author{Eugene Rabinovich}
\address{Department of Mathematics \\ University of California, Berkeley}
\email{erabin@math.berkeley.edu}

\date{}

\begin{document}
\maketitle
\begin{abstract}
Given a smooth family of massless free fermions parametrized by a base manifold $B$, we show that the (mathematically rigorous) Batalin-Vilkovisky quantization of the observables of this family gives rise to the determinant line bundle for the corresponding family of Dirac operators.
\end{abstract}
\tableofcontents
\section{Introduction}
In this article, we revisit the classical construction of the determinant line bundle of a family of Dirac operators \cite{quillen1, bf1, bf2} and show that it arises from the study of the mathematically rigorous theory of the massless free fermion \emph{in the Batalin-Vilkovisky formalism}. Recall that for a family of Dirac operators $D$ parametrized by a manifold $B$, one can produce a line bundle $L$ over $B$, known as the determinant line; here, we show that the BV formalism produces an infinite-rank graded vector bundle over $B$ whose sheaf of sections forms a complex of sheaves quasi-isomorphic to the sheaf of sections of $L$.

Let us explain the construction of $L$ in a bit more detail. The fiber of $L$ over the point $b\in B$ is
\[
(\Det D)_b: = \Lambda{}^{top} \ker D^+_b \otimes \Lambda{}^{top} (\coker D^-_b)^\vee,
\] 
where $D^+_b$ is the component of the Dirac operator at $b$ which acts on spinors of positive chirality, and the superscript $\vee$ denotes the dual vector space. We are assuming that $D_b$ acts on the space of sections of a vector bundle over a compact manifold $M_b$, so that $\ker D^+_b$ and $\coker D^+_b$ are finite-dimensional. \textit{A priori}, the lines $(\Det D)_b$ do not fit together into a smooth bundle over $B$ because the dimension of the kernel and cokernel of $D_b$ may jump as $b$ varies over $B$. However, a key insight of Quillen was that--for a family of $\bar\partial$ operators (i.e. holomorphic connections)--the lines $(\Det D)_b$ do in fact fit together into a smooth line bundle $\Det D$ over $B$. Moreover, let $B_0$ denote the subspace of $B$ consisting of all points $b$ such that $D_b$ has index zero, and $i: B_0\hookrightarrow B$ the inclusion. Then, one can construct a canonical section $\det D^+$ of $i^* \Det D$ as follows. If $D_b$ has a non-trivial kernel, then $(\det D^+)_b = 0$. Otherwise, since $D_b$ has trivial kernel and index zero, $(\Det D)_b$ is canonically identified with $\C$, so we set $(\det D^+)_b=1$ in such a case. Quillen showed that this is a smooth section of the line bundle $\Det D$.

This construction has a natural interpretation in terms of the quantum field theory of the free fermion. A Dirac operator defines a massless free fermion theory, and the path integrals for this theory are given heuristically by infinite-dimensional Berezin integrals. In particular, the value of the path integral computing the expectation value of the observable $1$ in the theory defined by $D_b$ should be given by $(\det D^+)_b$. This provides a field-theoretic interpretation for the section $(\det D^+)_b$ described in the previous paragraph. We note that, in the physics literature, the section $\det D^+$ is called the \textit{partition function}, but the terminology is a bit misleading since $\det D^+$ can be interpreted as an honest function only when $\Det D$ is provided with an explicit trivialization.

The path integral remains a heuristically-defined object; however, one approach to making sense of this object is the Batalin-Vilkovisky (BV) formalism, which replaces path integral computations with homological algebra. In the present work, we use the BV formalism as described in \cite{cost}, \cite{CG1}, and \cite{CG2}. The input to the BV formalism is a field theory and the output is a cochain complex of quantum observables of the theory. It is a feature of the BV formalism that the homological-algebraic context for the BV formalism can be quite general. One may construct cochain complexes not just over a field $k$ but over a suitable commutative ring, for example over the algebra of Chevalley-Eilenberg cochains of a Lie algebra $\fg$. In the present case, the homological-algebraic context is that of sheaves of $\cinfty_B$-modules. In particular, the quantum observables will be a(n infinite-rank) graded vector bundle $\Obq$ over $B$, whose sheaf of sections we will denote $\shObq$. $\Obq$ is endowed with a fiberwise differential, so that $\shObq$ is a complex of sheaves of $\cinfty_B$ modules. The main result of this note is the following theorem:

\begin{theorem}
\label{thm: main}
There is a quasi-isomorphism of sheaves of $\cinfty_B$-modules
\begin{equation*}
\Phi: \shObq \to \Gamma(\cdot, \Det D),
\end{equation*}
where $\Gamma(\cdot, \Det D)$ is the sheaf of sections of the line bundle $\Det D\to B$. Moreover, $\cinfty_B\subset \shObq$ and $\Phi(1) = \det D^+$.
\end{theorem}
In other words, we show that the BV formalism provides a mathematically precise justification for the physical origin of the determinant line bundle. Families of free fermions in the BV formalism have already been studied in \cite{rabaxial} in the case that the space $B$ is $B\mathfrak g$, where $\fg$ is a dgla; the aim of the present work is to adapt the methods of that paper to the context studied by Bismut, Freed, and Quillen. In other words, we seek to turn the physical inspiration for the constructions of \cite{bf1}, \cite{bf2}, \cite{quillen1} into precise mathematics. We will see that the suite of homological-algebraic techniques from \cite{rabaxial} suffice for the proof of Theorem \ref{thm: main}.
 
There are a number of further interpretations of Theorem \ref{thm:main}. Physically, $\det D^+$ is interpreted as the partition function of the theory. Theorem \ref{thm: main} states that $\Phi(1)$ is the partition function. More generally, one should understand the map $\Phi$ as computing (unnormalized) expectation values of observables---we comment on this in more detail below in Remark \ref{rem: vevs}. Moreover, using the techniques of \cite{CG1}, one can show that each fiber $\Obq_b$ is the space of global sections of a local-to-global object known as a \textit{factorization algebra}. Our Theorem therefore manifests the determinant line bundle as the global  sections of a local-to-global object. See Remark \ref{rem: factdet} for more details.
\subsection{Future Directions}
In \cite{quillen1}, Quillen also constructed a natural metric and compatible connection on $\Det D$ that could be used to study its characteristic classes via Chern-Weil theory. If the curvature and holonomy of this connection vanish, then a horizontal section gives a canonical (up to overall scale) trivialization of $\Det D$ which can be used to turn $\det D^+$ into an honest function on $B_0$. Bismut and Freed generalized these results to a wider class of elliptic operators. It would be interesting to know whether these objects can be induced from analogous structures on $\Obq$. 

\subsection{Organization of the paper}
 In section \ref{sec: fin-dim}, we study the finite-dimensional massless free fermion, first in the case $B=pt$, and then in the general case. We explore the relationship of the massless free fermion to the Pfaffian of a skew-symmetric bilinear form. In section \ref{sec: detbund}, we review the determinant bundle construction. The tools we discuss will enable us to, in Section \ref{sec: inf-dim}, prove Theorem \ref{thm: main}. The main homological algebraic tool throughout will be the homological perturbation lemma and related results, which we assemble in Appendix \ref{sec: app}.
\subsection{Related Work}
We are aware of related work \cite{mnev, schwarz}. The field theory we consider is a special case of the one introduced in \cite{schwarz}; however, our work differs from both of these references in that its focus is on \emph{families} of Dirac operators and the resulting bundles of \textit{observables}. By contrast, in the references \cite{mnev, schwarz}, the authors study the \emph{states} (and in particular the partition function) \emph{of a single theory}.

\subsection{Acknowledgements}
The author would like to thank O. Gwilliam, M. Ludewig, and P. Teichner for helpful discussions. He would also like to thank O. Gwilliam and B. Albert for helpful comments on an earlier draft of the note.

\subsection{Conventions}
\begin{itemize}
\item Throughout, the ground field is $\C$.
\item We use the notation $\cong$ for isomorphisms and $\simeq$ for quasi-isomorphisms (weak equivalences). 
\item We use the notation $\cinfty_M$ and $\Omega^k_M$ for the sheaves of smooth functions and smooth $k$-forms, respectively, on $M$. We use the notation $\underline{\R}$, $\Lambda^k T^*M$ for the corresponding vector bundles.
\item If $V$ is a finite-dimensional vector space, $V^\vee$ is the linear dual to $V$; if $V$ is a topological vector space, then $V^\vee$ is the continuous linear dual. If $V\to M$ is a finite-rank vector bundle, then $V^\vee \to M$ is the fiberwise dual to $V$. 
\item If $V_1\to M_1$ and $V_2\to M_2$ are vector bundles and $M_3$ is a manifold, then we denote by $V_1\boxtimes V_2$ the bundle $p_1^*(V_1)\otimes p_1^*(V_2)\to M_1\times M_2\times M_3$, where $p_1: M_1\times M_2\times M_3 \to M_1$ and $p_2: M_1\times M_2\times M_3\to M_2$ are the canonical projections. If $M_3$ is not explicitly mentioned, it is assumed to be $pt$.
\item If $V$ is a $\Z\times \Z/2$-graded vector space and $v$ is a homogeneous element of $v$, then $|v|$ and $\pi_v$ refer to the $\Z$ and $\Z/2$-degrees of $v$, respectively. Our sign convention is that moving a $v$ past $w$ incurs a sign of 
\[
(-1)^{|v||w|+\pi_v\pi_w}.
\]
\item If $V$ is a $\Z$-graded vector space, then $V[k]$ is the $\Z$-graded vector space whose $i$-th homogeneous space is $V_{i+k}$. In other words, $V[k]$ is $V$ shifted \textit{down} by $k$ ``slots''. If $V$ is a vector space, one can view it as a $\Z$-graded vector space concentrated in degree 0, and define $V[k]$ similarly.
\item If $U$ is a manifold, we will use the term ``sheaf of $\cinfty$ modules'' interchangeably with the term ``module for the sheaf of rings $\cinfty_U$''.
\item If $B$ is a manifold and $V\to B$ is a bundle over $B$, we denote by $V_b$ the fiber of $V$ over $b\in B$. 
\end{itemize}

\section{Finite-dimensional fermionic theories}
\label{sec: fin-dim}
In this section, we outline the relationship of the finite-dimensional massless free fermion to the Pfaffian of a skew-symmetric bilinear form $A$ on a vector space $W$. Aside from introducing many of the tools that we will use in proving Theorem \ref{thm: main}, this section serves as justification for the idea that the partition function of fermionic is computed by a Pfaffian. To the best of the author's knowledge, the exposition of this section---as well as the articulation and proof of Lemma \ref{lem: pfaffiso}---are original, though the connection between 0-dimensional fermionic theories and Pfaffians is well-known (see, e.g., Chapter 1.7 of \cite{zinnjustin}).

\subsection{The Pfaffian of a skew-symmetric bilinear form}
As in the preceding paragraph, let $A$ be a skew-symmetric bilinear form on a finite-dimensional vector space $W$ over  $\C$. Throughout, we use $A$ to denote both the corresponding element of $\Lambda^2 W^\vee$ and the corresponding skew-self-adjoint linear map $W\to W^\vee$. 
\begin{definition}
Given a skew-symmetric bilinear form $A$ on $W$, the \textbf{Pfaffian} of $A$, denoted $\pf(A)$, is the image of $e^A \in \Lambda^* W^\vee$ under the projection
\[
\pi_{top}:\Lambda^* W^\vee \to \Lambda^{top} W^\vee.
\]
We adopt the convention that if $W=0$, $\Lambda^{top}W^\vee:=\C$.
\end{definition}

From this definition, a number of properties are immediate. 
\begin{lemma}
Let $A$, $A'$ be skew-symmetric bilinear forms on $W$, $W'$ respectively. Then,
\begin{enumerate}
\item If $A$ is degenerate (i.e. if, considered as a map $W\to W^\vee$, $A$ is not invertible), then $\pf(A)=0$. 
\item Under the isomorphism $\Lambda^{top}((W\oplus W')^\vee)\cong \Lambda^{top}(W^\vee)\otimes \Lambda^{top}((W')^\vee))$, $\pf(A\oplus A')$ corresponds to $\pf(A)\otimes \pf(A')$.
\end{enumerate}
\end{lemma}

\begin{remark}
\label{rem: pfafftodet}
We note that if $T: W_0\to W_1$ is any linear map, then $\check{T}:=T\oplus (-T^\vee)$ defines a skew-symmetric map $W\to W^\vee$, where $W: = W_0\oplus W_1^\vee$. In this case, 
\[\ker \check T= \ker T \oplus \left(\coker T\right)^\vee,\] 
\[
\Lambda^{top} W^\vee = \Lambda^{top} W_0^\vee \otimes \Lambda^{top} W_1,
\] 
and
\[
\Lambda^{top} (\ker \check T )^\vee= \Lambda^{top} (\ker T)^\vee \otimes \Lambda^{top} \coker T.
\]
If, further, $\dim W_0=\dim W_1=:n$, then the Pfaffian $\pf(\check T)$ coincides with $(-1)^{n(n-1)/2}\det T\in \Lambda^{top} W_0^\vee \otimes \Lambda^{top} W_1$. 
\end{remark}

There is also the notion of \emph{Pfaffian homomorphism: }

\begin{definition}
Choose a splitting $W=\ker A\oplus K$. This induces an inclusion $(\ker A)^\vee\hookrightarrow W^\vee$. The \textbf{Pffafian homomorphism} corresponding to this splitting is the composite
\[
\begin{tikzcd}
\text{Pf}(A): \Lambda^{top} (\ker A)^\vee \ar[r,hook]& \Lambda^{*}W^\vee \ar[r,"\cdot e^A"]&\Lambda^{*}W^\vee\ar[r]&\Lambda^{top}W^\vee,
\end{tikzcd}
\]
where the last arrow is simply projection onto the appropriate direct summand, and the penultimate map is multiplication by $e^A$.
\end{definition}

As stated, the Pfaffian homomorphism depends on the choice of splitting. However, the next lemma shows that this is not the case.

\begin{lemma}
The Pfaffian homomorphism does not depend on the choice of splitting.
\end{lemma}
\begin{proof}
Let $W=\ker A\oplus K$ be a splitting, and let $r=\dim\ker A$, $s=\dim K$, $n=r+s$. Note that, under this decomposition, we can view
\begin{equation}
\label{eq: simplifyA}
A\in \Lambda^2 K^\vee,
\end{equation}
since $A$ is skew-symmetric and vanishes when one or another of its inputs is in $\ker A$. A different splitting of $W$ can be represented by a block upper-triangular linear transformation 
\[
S:=\left(\begin{matrix}
id_{\ker A} & B\\
0 & C
\end{matrix}\right),
\]
where $B: K\to \ker A$ is any linear transformation and $C\in \Aut(K)$. It follows that, with respect to the fixed decomposition using $K$, the map 
\[
S^\flat:\Lambda^r(\ker A)^\vee \hookrightarrow \Lambda^rW^\vee
\]
arising from the new splitting is induced by the linear map 
\[
id_{\ker A^\vee}\oplus B^\vee: (\ker A)^\vee \to (\ker A)^\vee \oplus K^\vee \cong W^\vee.
\]
Writing 
\[
\Lambda^r W^\vee \cong \oplus_{i=0}^r \Lambda^i (\ker A)^\vee \otimes \Lambda^{r-i} K^\vee,
\]
$S^\flat$ has, in general, image in all summands of the above direct sum. Given $\omega\in \Lambda^r (\ker A)^\vee$, write $S^\flat(\omega)=\sum_{i=0}^r \omega_i,$ where $\omega_i \in  \Lambda^{r-i} (\ker A)^\vee \otimes \Lambda^{i} K^\vee$. In this notation, $\omega_0$ corresponds to the image of $\omega$ in the original map $\Lambda^r (\ker A)^\vee\to \Lambda^r W^\vee$. Note that 
\[
e^A\omega_i \in \Lambda^{r-i}(\ker A)^\vee \otimes \Lambda^* K^\vee,
\]
because of Equation \ref{eq: simplifyA}. Therefore, the only term in $e^A S^\flat(\omega)$ that could possibly survive under the projection $\Lambda^* W^\vee \to \Lambda^n W^\vee\cong \Lambda^r (\ker A)^\vee\otimes \Lambda^s K^\vee$ is $e^A\omega_0$, which is precisely the image of $\omega$ under the Pfaffian homomorphism corresponding to $K$.
\end{proof}

By virtue of the previous lemma, we can, without ambiguity, use the term ``the Pfaffian homomorphism.'' We will prove the following lemma in the next subsection. The proof uses the Batalin-Vilkovisky formalism.

\begin{lemma}
\label{lem: pfaffiso}
The Pfaffian homomorphism is an isomorphism.
\end{lemma}

\subsection{The Pfaffian in the BV formalism}
In this subsection, we describe the Batalin-Vilkovisky (BV) formalism as applied in the case of interest to us. In the interest of brevity, we will not discuss the general framework of the BV formalism; instead, we will simply present the chain complexes of observables the formalism produces. However, we wish to emphasize that these chain complexes arise not by any methods particular to the problem at hand, but by the very general methods of the BV formalism.  For a more general discussion of the BV formalism as applied to the massless free fermion, see \cite{rabaxial}. We refer to \cite{CG1} for a discussion of the BV formalism as applied to general free field theories. 

Just as in the previous subsection, we let $W$ be a vector space, and $A$ a skew-symmetric pairing on $W$. Then, we let $V$ denote the cochain complex in super-vector spaces.
\[
V:=
\Pi\left(\begin{tikzcd}
W\ar[r,"A"]& W^\vee[-1],
\end{tikzcd}\right)
\]
with the symbol $\Pi$ denoting that the whole complex is purely odd in the $\Z/2$ grading. The $\Z/2$ grading is meant to track the underlying particle statistics---odd for fermionic statistics and even for bosonic statistics---and the additional, cohomological grading on $V$ will further affect the commuting or anti-commuting nature of the corresponding  observables. With respect to this latter grading, the notation indicates that $\Pi W$ lives in cohomological degree 0 and $\Pi W^\vee$ in cohomological degree 1. We let $\ip$ denote the $-1$-shifted pairing on $V$ induced from the evaluation pairing $W^\vee\otimes W\to \C$. This pairing looks symmetric, but is actually graded skew-symmetric in the sense that it is skew-symmetric once appropriate Koszul sign rules for $\Z\times \Z/2$-graded objects are taken into account (see Remark \ref{rem: koszul}).

Let $V^\sharp$ denote the underlying $\Z$-graded vector space of $V^\vee$; then, we let
\[
\alg : = \Sym(V^\sharp).
\]

\begin{remark}
\label{rem: koszul}
The symmetric algebra is taken with respect to the Koszul sign rule for $\Z\times \Z/2$-graded objects. More precisely, when moving an object of $\Z$-grading $d_1$ and $\Z/2$-grading $\pi_1$ past an object with $\Z$-grading $d_2$ and $\Z/2$-grading $\pi_2$, we obtain a sign 
\[
(-1)^{d_1d_2+\pi_1\pi_2}.
\]
In the symmetric algebra $\alg$, this means that elements of $W$ (which live in degree --1) commute with each other and anti-commute with elements of $W^\vee$, which live in degree 0 and anti-commute with each other.
\end{remark}

We will be interested in several differentials on $\alg$; the first is the one which is simply induced by extending the differential on $V^\vee$ to be a derivation on $\alg$.
\begin{definition}
The complex of \textbf{classical BV observables} of the theory described by $V$ and $\ip$ is the cochain complex
\[
\Obcl:=\Sym(V^\vee).
\] 
In other words, $\Obcl$ has $\alg$ as its underlying vector space and $A^\flat$ as its differential. Here, we use $A^\flat$ to refer to the differentials on $V^\vee$ and $\Sym(V^\vee)$ induced from that on $V$.
\end{definition}

\begin{remark}
Note that $V^\vee= V[1]$, and the differential on $V^\vee$ is $-A$, because of the skew-symmetry of $A$. The differential on $\sO(V)=\Sym(V[1])$ is then induced from $-A$ by the requirement that the differential be a derivation.
\end{remark}

\begin{lemma}
\label{lem: classdefretract}
There exists a deformation retraction 
\[
\begin{tikzcd}
\left( \Sym((H^\bullet V)^\vee), 0\right )\arrow[r,shift left = .5 ex,"\tilde \iota"] &  (\Obcl,A^\flat)\arrow[l, shift left = .5 ex,"\tilde \pi"]\arrow[loop right, distance = 4em,start anchor = {[yshift = 1ex]east},end anchor = {[yshift=-1ex]east}]{}{\tilde \eta}
\end{tikzcd}
\]
of the classical observables onto the symmetric algebra of functions on the cohomology of $V$. 
\end{lemma}
\begin{proof}
Let us choose a metric (i.e. a Hermitian inner product) on $W$. This enables us to write 
\begin{equation}
\label{eq: decompofV}
V\cong H^\bullet V \oplus V_{\perp},
\end{equation}
where $V_\perp$ is acyclic. More explicitly, we identify $H^\bullet V$ in $V$ as $\ker A \oplus \left((\Im A)^\perp[-1]\right)$ (where the superscript $\perp$ indicates orthogonal complement with respect to the chosen metric). Moreover, 

\[
V_\perp:=
\begin{tikzcd}
(\ker A)^\perp\ar[r,"A"]& \Im A
\end{tikzcd}.
\]
Letting $\eta$ denote the operator on $V$ which is 0 on $H^\bullet V$ and which is $A^{-1}$ on $V_\perp$, we have a deformation retraction

\[
\begin{tikzcd}
\left( H^\bullet V, 0\right )\arrow[r,shift left = .5 ex,"\iota"] &  (V,A)\arrow[l, shift left = .5 ex," \pi"]\arrow[loop right, distance = 4em,start anchor = {[yshift = 1ex]east},end anchor = {[yshift=-1ex]east}]{}{\eta}
\end{tikzcd}.
\]
Here, $\iota$ and $\pi$ are the inclusion and projection maps afforded by the decomposition of $V$ in equation \ref{eq: decompofV}. Moreover, $\pi\iota = id_{H^\bullet V}$ and $\eta$ is a cochain homotopy between $\iota\pi$ and the identity on $V$. As described in Section 2.5.3 of \cite{othesis}, these data also give rise to the deformation retraction of the lemma.

\end{proof}

\begin{remark}
We have given more than just a computation of the cohomology of the classical observables. We have also given a specific pair of inverse quasi-isomorphisms between the cohomology of the classical observables and the full complex of observables, along with an explicit contracting homotopy. When $W$ is replaced by an infinite-rank vector space, the construction of the full deformation retraction will allow us to avoid the sort of subtleties which appear when one attempts to compute cokernels in abelian (or additive) categories of infinite-dimensional vector spaces.
\end{remark}

Now, let us turn to the quantum observables. The quantum observables have the same underlying cochain complex, but a differential which is deformed by an operator known as the \textit{BV Laplacian}. To define this operator, note that the pairing $\ip\in \Lambda^2 V^\vee$ is non-degenerate. This implies that it has an inverse $\ip^{-1}\in \Sym^2 V$ defined by the equation 
\[
-(id_V\otimes \ip) \left(\ip^{-1}\otimes v\right) = v.
\]
Using this equation, one can verify that $\ip^{-1}$ is indeed symmetric, though $\ip$ is anti-symmetric.
\begin{definition}
The \textbf{BV Laplacian}, denoted $\Delta$, is the unique operator on $\alg$ which is 0 on $\Sym^{\leq 1}(V^\sharp)$,  contraction with $\ip^{-1}$ on $\Sym^2(V^\sharp)$, and which is extended to the rest of $\alg$ as a second-order differential operator.
\end{definition}

\begin{remark}
To get a better feel for $\Delta$, choose a basis $\{e_i\}_{i=1}^n$ for $W$. Let $\{x_i\}$ denote the dual basis, thought of as a collection of elements of cohomological degree 0 in $V^\vee$. Similarly, let $\xi_i$ denote $e_i$, but thought of as an element of degree --1 in $V^\vee$. Then, 
\[
\alg  \cong \C[x_1, \cdots, x_n, \xi_1, \cdots, \xi_n]
\]
where the $x_i$ have $\Z$-degree 0 and anti-commute with each other, while the $\xi_i$ have $\Z$-degree --1 and commute with each other. The $x_i$'s and the $\xi_i$'s anti-commute with each other. Under this isomorphism, 
\[
\Delta = \sum_{i=1}^n \frac{\partial^2}{\partial x_i \partial \xi_i}.
\]
If $\{e'_i\}$ is another choice of basis for $W$ and $\{x'_i\}$, $\{\xi'_i\}$ are the correspondingly adjusted elements of $V^\vee$, then
\[
\Delta = \sum_{i=1}^n \frac{\partial^2}{\partial x'_i \partial \xi'_i}.
\]
In this sense, the characterization of $\Delta$ is independent of the choice of basis $\{e_i\}$.
\end{remark}

\begin{definition}
The \textbf{quantum observables} associated to $V, \ip$ are the cochain complex
\[
(\alg, A^\flat+\Delta).
\]
We denote these by $\Obq$. The \textbf{trivial quantum observables} associated to the pair $(V,\ip)$ are the cochain complex 
\[
(\alg, \Delta).
\]
We denote these by $\Obq_0$.
\end{definition}
\begin{remark}
We have not proven that $A^\flat+\Delta$ and $\Delta$ are differentials. We leave it as a straightforward exercise to check that $\Delta^2=0$ and $[A^\flat, \Delta]=0$.
\end{remark}

\begin{remark}
Though this isn't reflected in the notation, $\Obq$ depends on $A$; on the other hand, $\Obq_0$ is the cochain complex of quantum observables for the theory with the choice $A=0$.
\end{remark}

We now turn to the question of computing the cohomology of the (trivial) quantum observables. To begin with, note that, since $e^{A}\in \Lambda W^\vee$, and $\Lambda W^\vee\subset \alg$, we can think of $e^{A}$ as living in either $\Obq$ or $\Obq_0$. 
\begin{lemma}
\label{lem: multbypartfunc}
Multiplication by $e^{A}$ is a cochain isomorphism between $\Obq$ and $\Obq_0$. 
\end{lemma}
\begin{proof}
Direct computation.
\end{proof}

Let us now describe some results on the cohomologies of the (trivial) quantum observables. We refer the reader to Propositions 5.7, 5.8 and Corollary B of \cite{rabaxial} for the details. The proof of the result below is nearly identical to the aforementioned propositions of \cite{rabaxial}, so we do not show it here. A key component of the proof is the application of the homological perturbation lemma (Lemma \ref{lem: hpl}) to the deformation retraction of Lemma \ref{lem: classdefretract}.

\begin{proposition}
\label{prop: quantobs}
There are deformation retractions 
\[
\begin{tikzcd}
\left( \Lambda^{top} W^\vee ,0\right )\arrow[r,shift left = .5 ex," \iota'"] &  \Obq_0\arrow[l, shift left = .5 ex," \pi'"]\arrow[loop right, distance = 4em,start anchor = {[yshift = 1ex]east},end anchor = {[yshift=-1ex]east}]{}{ \eta'}
\end{tikzcd}
\]
and 
\[
\begin{tikzcd}
\left( \Lambda^{top} (\ker A)^\vee ,0\right )\arrow[r,shift left = .5 ex," \iota''"] &  \Obq\arrow[l, shift left = .5 ex," \pi''"]\arrow[loop right, distance = 4em,start anchor = {[yshift = 1ex]east},end anchor = {[yshift=-1ex]east}]{}{ \eta''}
\end{tikzcd}.
\]

Here, $\pi'$ is the projection 
\[
\alg = \bigoplus_{i,j=0}^\infty \Lambda^i W^\vee \otimes\Sym^j W \to \Lambda^{top} W^\vee.
\]
and $\iota''$ is the inclusion $\Lambda^{top} (\ker A)^\vee\hookrightarrow \alg$ arising from the choice of metric on $W$. Similarly, $\iota'$ is the inclusion $\Lambda^{top}W^\vee \hookrightarrow \Lambda^*W^\vee\hookrightarrow \sO(V)$.
\end{proposition}

\begin{remark}
\label{rem: berezin}
The BV formalism is a homological approach to integration. Because of the purely odd $\Z/2$ grading of $V$, the integration occurring in the example at hand is over the super-manifold $\Pi W$. $\Lambda^{*}W^\vee$ (in other words, the space of functions on $\Pi W$) is a subcomplex of $\Obq_0$, and the map $\pi'\mid_{\Lambda^*W^\vee}:\Lambda^* W^\vee \to \Lambda^{top}W^\vee$ is simply the Berezin integral. Moreover, the composition $\pi'\circ(\cdot e^A)$ computes Berezin integrals of functions on $\Pi W$ against the measure $e^A$, i.e. it is the map which computes expectation values in the theory determined by $A$, in physicists' terminology. Thus, the image of $1\in \Obq$ under $\pi'\circ (\cdot e^A)$ computes, on the one hand, the Berezin integral of $e^A$ (which physicists would call the partition function). On the other hand, this image is the Pfaffian of $A$, by definition. This is the sense in which the fermionic partition function is given by the Pfaffian of $A$.

We comment also that, unless $W$ carries extra data so that $\Lambda^{top}W^\vee$ is trivialized, expectation values of observables are not canonically numbers. This point becomes significant once $W\to B$ is allowed to be a general vector bundle over a base manifold $B$ (and so $\Lambda^{top}W^\vee$ may or may not be trivializable at all).
\end{remark}

Now, we can prove lemma \ref{lem: pfaffiso}.

\begin{proof}[Proof of lemma \ref{lem: pfaffiso}]
It follows directly from the definitions that the Pfaffian homomorphism is the map $\pi' \circ (
\cdot e^A)\circ \iota''$. Each of the three maps is a quasi-isomorphism, so the composite is also a quasi-isomorphism. A quasi-isomorphism between cochain complexes with 0 differential is simply an isomorphism, hence the lemma. 
\end{proof}

\begin{remark}
Notice that $\Obq$ and $\Obq_0$ are polynomial algebras, not power series algebras. $e^A$ is polynomial in $A$ thanks in part to the finite-dimensionality of $W$. In the infinite-dimensional case, however, $e^A$ is no longer polynomial in $A$. Thus, there will be no cochain isomorphism $\Obq_0\to\Obq$ in the infinite-dimensional case. Most of the techniques we have studied in this section will nevertheless apply.
\end{remark}

\subsection{Finite-dimensional families}
\label{subsec: findimfamilies}
Let now $B$ be a manifold, $W\to B$ a vector bundle on $B$, and $A:W\to W^\vee$ a skew-symmetric vector bundle endomorphism. We can define, as above, a $\Z$-graded vector bundle $V\to B$, and bundles $\sO(V)\to B$, $\Obcl\to B$, $\Obq\to B$. More precisely, $\sO(V)=\Sym(V^\vee)$, where the symmetric algebra and dual are taken in the vector bundle sense. We let $\shObcl$ denote the sheaf on $B$ of sections of $\Obcl$ and similarly for $\shObq$. The classical and quantum BV differentials endow $\Obcl$ and $\Obq$ with fiberwise differentials. These induce differentials on $\shObcl$ and $\shObq$, which endow these graded sheaves with the structure of complexes of sheaves of $\cinfty_B$ modules. 

We can apply the discussion of the previous section fiberwise: the cohomology of $\Obq_b$ is $\Lambda^{top}(\ker A_b)^\vee$, where $b\in B$ and the subscript $b$ denotes the fiber at $b$. However, it is not clear what sort of structure the fibers $\Lambda^{top}(\ker A_b)^\vee$ fit into as $b$ varies over $B$; the issue is that the rank of $\ker A_b$ might jump discontinously as $b$ varies. 

As discussed in the Introduction, it was the insight of Quillen to give $\Lambda^{top}(\ker A)^\vee$ the structure of a smooth vector bundle over $B$. Because we are in the finite-rank case, we have the (inverse of the) fiberwise Pfaffian isomorphism 
\[
\Lambda^{top}W_x^\vee \to \Lambda^{top} (\ker A_x)^\vee.
\]
We can view this isomorphism as \textit{defining} the smooth structure on $\Lambda^{top}(\ker A)^\vee$ in such a way that the Pfaffian isomorphism is a smooth bundle map. Moreover, just as above, we have a cochain isomorphism
\[
e^A: \shObq \to \shObq_0.
\]
Finally, for sufficiently small open subsets $U\subset B$, we can construct deformation retractions
\[
\defretract{\Gamma(U,\Lambda^{top}W^\vee)}{\shObq_0(U)}{\iota'}{\pi'}{\eta'}
\]
by performing the constructions of Proposition \ref{prop: quantobs} locally. By the characterization of $\iota'$ and $\pi'$ in Proposition \ref{prop: quantobs}, it is clear that $\iota'$ and $\pi'$ extend to well-defined maps of complexes of sheaves on $B$. In fact, $\iota'$ and $\pi'$ arise from bundle maps 
\[
\begin{tikzcd}
\Lambda^{top} W^\vee \ar[r, shift left = 0.5ex,"\iota'"]&\Obq_0\ar[l,shift left = .5ex,"\pi'"].
\end{tikzcd}
\] 
(From the explicit formula for $\eta'$ in \cite{rabaxial}, it is evident that $\eta'$ does not so extend, though it is enough for our purposes that $\eta'$ be locally well-defined.)  Hence, because of the above locally-constructed deformation retractions, we have shown the result of the following lemma:

\begin{lemma}
\label{lem: famquantobs}
The maps $\iota', \pi'$ in Proposition \ref{prop: quantobs} generalize to mutually inverse quasi-isomorphisms 
\[
\begin{tikzcd}
\Gamma(\cdot, \Lambda^{top}W^\vee) \ar[r,shift left= .5ex,"\iota'"]&\shObq_0\ar[l,shift left= .5ex,"\pi'"],
\end{tikzcd}
\]
where the maps are the natural inclusions and projections that arise after one observes that $\shObq_0=\bigoplus_{i,j=0}^\infty \Gamma(\cdot, \Lambda^i W^\vee \otimes \Sym^j W)$ (isomorphism of graded sheaves).
\end{lemma}

As before, multiplication by $e^A$ is a cochain isomorphism $\shObq_0\to \shObq$, and the Pfaffian isomorphism is by definition a smooth bundle map $\Lambda^{top}(\ker A)^\vee\to \Lambda^{top}W^\vee$.  Hence, we also have mutually inverse quasi-isomorphisms $\shObq\leftrightarrow \Lambda^{top}(\ker A)^\vee$.

\begin{remark}
In comparing Lemma \ref{lem: famquantobs} with Proposition \ref{prop: quantobs}, we find that the contracting homotopies cease to be globally defined once $B\neq pt$. This is not a serious issue, since the only purpose of the homotopies is as an auxiliary device in proving the fact that the $\iota'$ and $\pi'$ are quasi-isomorphisms. In Section \ref{sec: inf-dim} (see Remark \ref{rem: iotafails}), we will see that, in addition, the analogue of $\iota'$ will cease to be well-defined globally on $B$. Nevertheless, $\pi'$ will survive.
\end{remark}

\section{The determinant line bundle construction}
\label{sec: detbund}
\subsection{Families of Dirac operators}
\label{subsec: setup}
In this section, we give a rapid overview of the determinant bundle construction. All stated results in this section come from \cite{bgv}.

The fundamental datum used to define the massless free fermion in the BV formalism in \cite{rabaxial} was a formally self-adjoint Dirac operator $D$. For a brief outline of the definitions and results relevant to the theory of such operators, we refer the reader to section 2 of \cite{rabaxial}. In the case at hand, we would like to understand the necessary constructions in the case that one is presented with a family of Dirac operators.

Here, and throughout, let $\tau: M\to B$ be a smooth fiber bundle over the smooth manifold $B$ with compact fibers $\{M_b\}_{b\in B}$. Let $V\to M$ be a complex vector bundle with metric $(\cdot, \cdot)$ and an orthogonal decomposition $V=V^+\oplus V^-$. By $\sV$, we will denote the bundle over $B$ whose fiber at $b$ is $\Gamma(M_b, V_b)$. Smooth sections of $\sV$ are, by definition,
\[
\Gamma(U,\sV) := \Gamma(\tau^{-1}(U), V).
\]

Now, consider a smooth family of Dirac operators $D$ over $B$. By a smooth family of differential operators, we mean a collection $\{D_b\}_{b\in B}$, where $D_b$ is a differential operator on the bundle $V_b\to M_b$, such that the coefficients of the differential operator vary smoothly in $b$ in a coordinate expression for $D$ on a trivializing neighborhood for $V$ and $M$. We require each $D_b$ to be a formally self-adjoint Dirac operator on $V_{M_b}\to M_b$. In more detail $D_b$ consists of a pair of maps $D_b^\pm:\Gamma(M_b, V_b^\pm)\to \Gamma(M_b,V_b^\mp)$ such that $D_b^+D_b^-$ and $D^-_bD^+_b$ are generalized Laplacians for the same metric on $M_b$. In this situation, the bundle $\sV$ carries a canonical metric, which pairs sections using the inner product $(\cdot, \cdot)$ on $V$ and then integrates the resulting function against the Riemannian metric induced from $D_b$ over the compact fibers $M_b$.  The requirement that $D_b$ be formally self-adjoint translates into the requirement that $D_b^-$ be the formal adjoint of $D^+_b$ with respect to this metric on $\sV$.

\subsection{Spectral Cutoffs}
If $W_0$, $W_1$ are two finite-dimensional vector spaces of equal rank, and $T: W_0\to W_1$ is a linear operator, we can define $\det T: \Lambda^{top} W_0 \to \Lambda^{top} W_1$, which we can view equivalently as an element 
\[
\det T\in \Lambda^{top}W_0^\vee\otimes \Lambda^{top}W_1=:\Det T.
\]
When $W_0$ and $W_1$ are infinite-dimensional, $\Lambda^{top}W_0$ and $\Lambda^{top}W_1$ no longer make sense. One way out is to perform the following maneuver: in the finite-dimensional case, $T$  defines an isomorphism 
\[
(\Lambda^{top} W_0)^\vee\otimes \Lambda^{top}W_1\cong (\Lambda^{top} \ker T)^\vee\otimes \Lambda^{top} \coker T
\]
In the case we are interested in, $T$ is nevertheless Fredholm, so that $(\Lambda^{top} \ker T)^\vee\otimes \Lambda^{top} \coker T$ makes sense even when $(\Lambda^{top} W_0)^\vee\otimes \Lambda^{top}W_1$ does not. 

Hence, our candidate for the determinant bundle of a family of formally self-adjoint Dirac operators is $(\Lambda^{top}\ker D^+)^\vee\otimes \Lambda^{top}\coker D^+$. However, while this construction makes sense fiber-by-fiber over $B$, the dimension of $\ker D^+$ and $\coker D^+$ may jump and so, as mentioned in the second paragraph of the introduction, the smooth structure on this candidate is not manifest.

The solution is to introduce finite-dimensional bundles $K_\lambda$, one for each $\lambda\geq 0$, which are ``bigger'' than $\ker D^+$ and $\coker D^+$ but can nevertheless be used to define a determinant line $\Det D_\lambda$. $\Det D_\lambda$ is fiberwise isomorphic to $(\Lambda^{top}\ker D^+)^\vee\otimes \Lambda^{top}\coker D^+$, but $K_\lambda$ will be constant rank, so that $\Det D_\lambda$ will have a manifest smooth structure. The drawback of the bundles $K_\lambda$, however, is that they are defined only over open subsets  $U_\lambda$ of $B$, and so we will have to define transition functions $\Det D_\lambda \to \Det D_\mu$  on the overlaps $U_\lambda\cap U_\mu$.

Let us now flesh out the details of this approach. Let $\lambda\geq 0$, and define $U_\lambda$ to be the subset of $B$ on which $\lambda$ is not an eigenvalue of $D^2$. Given $b\in B$, let $P_{[0,\lambda)}^b$ be the family of orthogonal projectors onto the space spanned by the eigenvectors of $(D_b)^2$ with eigenvalue in $[0,\lambda)$, and define $P_{[0,\lambda)}$ to be the family $(P_{[0,\lambda)}^b\mid b\in B)$.  The following is proposition 9.10 in \cite{bgv}:
\begin{proposition}
\label{prop: spectralprojectors}
\begin{enumerate}
\item The sets $U_\lambda$ are open, and they cover $B$.
\item The family $P_{[0,\lambda)}$ is a smooth family of smoothing operators on $U_\lambda$. 
\item The vector spaces $(K_{\lambda})_b:=\Im(P^b_{[0,\lambda)})$ combine into a smooth finite-rank vector subbundle $K_\lambda$ of $\sV$ over $U_\lambda$. $K_\lambda$ possesses an orthogonal decomposition $K_\lambda^+\oplus K_\lambda^-$ which agrees with the original one on $V$.
\end{enumerate}
\end{proposition}

Here, to say that $P_{[0,\lambda)}$ is a smooth family of smoothing operators means that there is a section of $V^{\boxtimes 2}$ on $B\times B$ such that convolution with this section gives the family of operators $P_{[0,\lambda)}$.

Similarly, we can define the projector $P_{(\lambda,\mu)}$ and vector bundle $K_{(\lambda,\mu)}$ over $U_\lambda\cap U_\mu$, and we have the following relations
\[
P_{[0,\mu)}=P_{[0,\lambda)}+P_{(\lambda,\mu)},
\]
\[
K_\mu=K_\lambda\oplus K_{(\lambda,\mu)},
\]
where the direct sum is orthogonal and the equalities hold only over $U_\lambda\cap U_\mu$. 

We denote by $D_\lambda$ the family of operators $D$ restricted to $K_\lambda$, and similarly for $D_\mu$ and $D_{(\lambda,\mu)}$. One evidently has the relation
\[
D_\mu = D_\lambda +D_{(\lambda,\mu)}
\]
over $U_\lambda\cap U_\mu$.
\subsection{The determinant line bundle and the determinant section}
With the finite- and constant-rank superbundles $K_\lambda$ in hand, we can finally define the determinant line bundle. To this end, let $K^\pm_\lambda$ be the components of $K_\lambda$ in the decomposition of Proposition \ref{prop: spectralprojectors}, and define 
\[
\Det(K_\lambda) := \Lambda^{top}(K^+_\lambda)^\vee\otimes \Lambda^{top}K^-_\lambda,
\]
Note that, on the overlaps $U_\lambda\cap U_\mu,$
\[
\Det(K_\mu)\cong \Det(K_\lambda)\otimes \Det(K_{(\lambda,\mu)}),
\]
and that $D$, when restricted to $K_{(\lambda,\mu)}$ is invertible, so that $\det(D_{(\lambda,\mu)}^+)$ defines an invertible element in $\Det(K_{(\lambda,\mu)})$, and so an isomorphism $\Det K_\mu\cong \Det K_\lambda$.

\begin{definition}
\begin{enumerate}
\item Let $K^+_\lambda$ and $K^-_\lambda$ denote the even and odd components of the super vector bundle $K_\lambda$. Over $U_\lambda$, set
\[
\Det(D)\mid_{U_\lambda}=\Det(K_\lambda)
\]
with transition functions
\[
\begin{tikzcd}
\Det(K_\lambda)\ar[rr,"\epsilon\cdot id\otimes \det(D_{(\lambda,\mu)}^+)"]&&\Det(K_\lambda)\otimes \Det(K_{(\lambda,\mu)})\cong \Det(K_\mu)
\end{tikzcd}
\]
(defined over $U_\lambda\cap U_\mu$), where $\epsilon = (-1)^{\rk K_{(\lambda,\mu)}^+(\rk K_{(\lambda,\mu)}^+-1)/2}$. It is straightforward to check that the transition functions satisfy the cocycle property on the triple overlaps. $\Det(D)$ is the \textbf{determinant line bundle} for the family of Dirac operators $D$.
\item Suppose $K_\lambda^+$ and $K_\lambda^-$ have equal ranks for any value of $\lambda$ (which happens if and only if $(D^+)_x$ has index 0 for all $x\in U_\lambda$). The \textbf{determinant section} $\det(D^+)\in \Gamma(B,\Det(D))$ is given by 
\[
(-1)^{\rk K_\lambda^+(\rk K_\lambda^+-1)/2}\det(D^+_\lambda)\in \Det(K_\lambda)
\]
on each open $U_\lambda$. So constructed, $\det(D^+)$ is equivariant with respect to the transition functions on the overlaps $U_\lambda\cap U_\mu$. The sign factor in the formulas is intended to account for the sign which appears when identifying $\det (T\oplus S)$ with $\det(T)\oplus \det (S)$ under the isomorphism $\Det T\otimes \Det S\cong \Det(T\oplus S)$.
\end{enumerate}
\end{definition}
So we have now defined half of the objects appearing in Theorem \ref{thm: main}. We turn now to a discussion of the massless free fermion in the Batalin-Vilkovisky formalism.

\section{The generalization to infinite dimensions}
\label{sec: inf-dim}
In this section, we prove Theorem \ref{thm: main}. To do so, we still have to define the sheaf $\shObq$ on $B$, which we do in subsection \ref{subsec: quantumobservablesbundle}.

\subsection{The bundles of BV observables of the massless free fermion in the infinite-dimensional case}
\label{subsec: quantumobservablesbundle}
We recall now the notation from Section \ref{subsec: setup}; $D$ is a family of Dirac operators parametrized by $B$, and $\sV$ is the infinite-rank bundle over $B$ on which $D$ acts. Let 
\[
\sS:= 
\Pi\left(\begin{tikzcd}[row sep = tiny]
\mathscr{V}^+[1] \ar[r,"D^+"] & \mathscr{V}^-\\
\oplus &  \oplus\\
(\mathscr{V}^{-})^{!}[1]\ar[r,"-(D^+)^!"]& (\sV^{+})^{!}
\end{tikzcd}\right),
\]
i.e. $\sS$ is a chain complex of super-vector bundles over $B$ concentrated in purely odd $\Z/2$-degree. Here, $\sV^!$ denotes the bundle over $B$ whose fiber over a point $b\in B$ is \[\Gamma(M_b, V^\vee_b)\] (with smooth structure defined analogously to that of $\sV$). Let
\[
\Obcl := \Sym(\sS),
\]
where here, $\Sym=\oplus \Sym^i$ and $\Sym^i$ is the operation which takes the fiberwise $i$-th completed projective tensor product of nuclear Fr\' echet spaces. More precisely, the fiber of $\Sym^i(\sS)$ over $b\in B$ is the subspace of 
\[
\Gamma(\overbrace{M_b\times \cdots\times M_b}^{i\text{ factors}},\overbrace{S_b\boxtimes \cdots \boxtimes S_b}^{i\text{ factors}})
\] 
consisting of sections which are symmetric under the natural $\Sigma_i$ action, where 
\[
S_b = (V^+)_b[1]\oplus (V^+)^!_b\oplus (V^-)^!_b[1]\oplus (V^-)_b
\] 
We note also that the symmetrization is taken with respect to the Koszul sign rules. $\Sym(\sS)$ has a differential which is induced from that on $\sS$ by the usual Leibniz rule. $\sS$ has a degree $+1$ pairing $\ip$ defined by the formula
\[
\ip[\varphi,\psi]=\int_{M_b}\varphi(\psi)dVol_{g_b}
\]
for $\varphi\in (\sV^-)^!_b=\Gamma(M_b,S_b^\vee),\psi \in \sV^-_b=\Gamma(M_b,S_b)$, and similarly for other combinations of fields.

 We define $\Delta$ to be the unique degree $+1$, second-order operator on $\Sym(\sS^\sharp)$ which is zero on $\Sym^{<2}$ and which is $\ip$ on $\Sym^2$. Then, define 
\[
\Obq:= \left(\Sym(\sS^\sharp),-\tilde D+\Delta\right),
\]
where $-\tilde D$ is simply the differential from $\Obcl$. We let $\shObq$ and $\shObcl$ denote the sheaves on $B$ of sections of $\Obq$ and $\Obcl$. Our ultimate goal is to show that there exists a quasi-isomorphism $\shObq\to \Det(D^+)$.

\begin{remark}
The most direct generalization of the discussion of the previous subsection would be to take 
\[
\Obcl:= \Sym(\sS'),
\]
where
\[
\sS':= 
\left(\Pi\left(\begin{tikzcd}[row sep = tiny]
\sV^+ \ar[r,"D^+"] & \sV^-[-1]\\
\oplus & \oplus\\
(\sV^-)^\vee \ar[r,"-(D^+)^\vee"]& (\sV^+)^\vee[-1]
\end{tikzcd}\right)\right)^\vee,
\]
where $(\sV^+)^\vee$ denotes the bundle over $B$ which is the fiberwise strong topological dual to $\sV^+$; however, this would introduce analytic difficulties that we can avoid by using $\sS$. In essence, to obtain $\sS$ from $\sS'$, we replace all occurences of the symbol ``$\vee$'' with the symbol ``$!$''.
\end{remark}

\subsection{Applications}
Having established all the relevant notation, we would like to discuss some important interpretations of Theorem \ref{thm: main} before embarking on its proof.
\begin{remark}
\label{rem: vevs}
The differential on $\Obq$ serves the purpose of encoding the fact that the path-integral of a total derivative is zero; an element in $\Obq_b$ is exact precisely if it is the divergence of a vector field on the space of fields (see Section 2.2. of \cite{CG1}). Theorem \ref{thm: main} states that, once one sets to zero all observables in $\Obq_b$ whose expectation value is zero for this reason, all observables are multiples of a single observable. $\Phi$ should therefore be interpreted as the map which computes expectation values, although it should be noted that unless $\Det D$ is trivial, expectation values are not canonically functions on $B$, but sections of $\Det D$. 

If, moreover, one wishes to study symmetries of the massless free fermion, one must introduce (background or dynamical) gauge fields to couple to the original theory. The differential on $\Obq_b$ will be modified by a term that imposes, in addition, Ward identities for the relevant symmetry. If one calls ``conformal blocks'' the space of all observables modulo the relations induced by Ward identities and the vanishing of total derivatives, then Theorem \ref{thm: main} states that the conformal blocks of the massless free fermion are one-dimensional, even without the imposition of any additional symmetries.
\end{remark}

\begin{remark}
\label{rem: factdet}
For each $b\in B$, $\Obq_b$ is the space of global sections of a factorization algebra on $M_b$ (see \cite{CG1}). A factorization algebra is a cosheaf-like local-to-global object on $M_b$ which satisfies a descent property for a particular class of covers. Another interpretation of Theorem \ref{thm: main}, therefore, is that we have ``factorized the determinant line'' in the sense that we have identified the line $(\det D)_b$ with $\Obq_b$ in the derived category of chain complexes  and that, moreover, $\Obq_b$ can be computed as a homotopy colimit of a diagram involving only local information on $M_b$. These identifications are smooth in $b\in B$. 
\end{remark}

\subsection{Proof of main theorem}
\label{subsec: proof}
We would like to give a quasi-isomorphism $\shObq\to \Det D$ along the lines of the previous section. However, the trick that enabled us to do this, namely the identification of $\Lambda^{top} W^\vee$ with $\Lambda^{top}(\ker A)^\vee$, no longer applies, since $\Lambda^{top} W^\vee$ is ill-defined when $W$ has infinite-rank. However, we will find that, over the open subsets $U_\lambda$ of $B$, we can reduce the problem to the finite-rank case. We will then have quasi-isomorphisms $\pi_\lambda: \shObq\to \Det D$; the last step will be to verify that those locally-defined quasi-isomorphisms are compatible with the transition functions for $\Det D$.  

Let us now flesh out this argument in more detail. 

We note first a technical lemma which will be important to us:

\begin{lemma}
Let $H_\lambda$ be the subvector-bundle of $\sV^!$ produced by the same means as $K_\lambda$, except using the Dirac operator $D^!$. Then $H_\lambda \cong K_\lambda^\vee$. 
\end{lemma}

\begin{proof}
First let us show that $(D^!)^2_b$ and $D^2_b$ have the same eigenvalues, so that $H_\lambda$ and $K_\lambda$ are both defined over $U_\lambda$. First, note that given an eigenvector $v_\mu\in \Gamma(M_b, V_b)$ with eigenvalue $\mu$ of $D_b^2$, then the metric on $V_b$ gives a section of $V_b^\vee$ which is an eigenvector of $(D^!)_b^2$ with eigenvalue $\mu$. The converse argument also holds to give an eigenvector for $D_b^2$ given one for $(D^!)_b^2$. 

Now, over $U_\lambda$, there is a natural smooth map 
\begin{equation}
\label{eq: bundleiso}
H_\lambda \subset \sV^!\to K_\lambda^\vee;
\end{equation}

more precisely, there is a map of $\cinfty_B(U)$ modules
\[
\Gamma(U,H_\lambda) \to \Hom_{\cinfty_{B}(U)}(\Gamma(U,K_\lambda), \cinfty_B(U)),
\]
which takes a section of $H_\lambda$, views it as a section of $\sV^!$, and then restricts this to a section of $K_\lambda^\vee$. The map manifestly intertwines the $\cinfty_B$ actions and is linear over $\cinfty_B$. 

It remains only to check that this map is a fiberwise isomorphism. In the first paragraph of the proof, we constructed an (anti-linear) isomorphism 
\[
(H_\lambda)_b \to (K_\lambda)_b;
\]
moreover, there is another anti-linear isomorphism arising from the pairing $(\cdot,\cdot)$ on $V_b$
\[
(K_\lambda)_b\to (K_\lambda)_b^\vee.
\]
We claim that the composite map $(H_\lambda)_b\to(K_\lambda)_b^\vee$ is the map of Equation \ref{eq: bundleiso}. To this end, note the following commutative diagram

\[
\begin{tikzcd}
\sV^!\ar[r]& \sV \ar[r]& \sV^!\ar[d]\\
H_\lambda \ar[u]\ar[r] &K_\lambda\ar[r,"\cong"]\ar[u,hook] & K_\lambda^\vee
\end{tikzcd};
\]
the left square commutes by the discussion of the first paragraph of the proof, and the right square commutes by the construction of the bottom map in that square. The composite of the top two rightward-pointing arrows is the identity, and so the map $H_\lambda\to K_\lambda^\vee$ which arises from going up and across the top row is the map of Equation \ref{eq: bundleiso}. On the other hand, the bottom map is an isomorphism; the lemma follows.
\end{proof}

Now, as in \cite{rabaxial}, we start by studying the classical observables:

\begin{lemma}
Let $U_\lambda$ denote the open subset of $B$ consisting of the points $x$ where $\lambda>0$ is not an eigenvalue of $D^2_x$, and let 
\[
\tilde K_\lambda := K_\lambda^-\oplus H_\lambda^+
\]
Then, there is a deformation retraction
\[
\defretract{\left(\Gamma(\cdot,\Sym(\Pi(\tilde K_\lambda\oplus \tilde K_\lambda^\vee[1]))),D^+-(D^+)^!\right)}{\shObcl}{\iota_\lambda}{\pi_\lambda}{\eta_\lambda}
\]
of sheaves of $\cinfty$-modules over $U_\lambda$, where $\iota_\lambda$, $\pi_\lambda$, and $\eta_\lambda$ are induced from smooth bundle maps. Here, we make a slight abuse of notation and denote by $D^+-(D^+)^!$ the operator on $\Sym(\Pi(\tilde K_\lambda\oplus \tilde K_\lambda\vee[1]))$ which extends $D^+-(D^+)^!$ as a derivation from the linear summand of the symmetric algebra to the whole space.
\end{lemma}
\begin{proof}
Let us denote by $\tilde D$ the operator $D^+-(D^+)^!$, and by $\check D$ the operator $D^--(D^-)^!$. Now, we construct a deformation retraction 
\begin{equation}
\label{eq: firstdefretract}
\defretract{(\Gamma(\cdot, \Pi(\tilde K_\lambda\oplus \tilde K_\lambda^\vee[1])), \tilde D)}{\Gamma(\cdot, \sS)}{\iota_\lambda}{\pi_\lambda}{\eta_\lambda};
\end{equation}
the map $\iota_\lambda$ is induced from the inclusion of $\Pi(\tilde K_\lambda \oplus \tilde K_\lambda^\vee[1])$ as a subbundle of $\sS$. The map $\pi_\lambda$ is induced from the smooth family of projections $\tilde P_\lambda:\sS\to \tilde K_\lambda\oplus \tilde K_\lambda^\vee[-1]$.  (This family is smooth by Proposition \ref{prop: spectralprojectors}.) Finally, let $(G_x)_{x\in B}$ be the family of Green's operators for the family of operators $\tilde D\check D+\check D\tilde D$. For $\eta_\lambda$ we take $ \check D G(1-\tilde P_\lambda)$, viewed as a degree --1 fiberwise endomorphism of $\sS$. Proposition 9.20 of \cite{bgv} guarantees that $\check D G$ is a smooth section of the endomorphism bundle of $\sS$, and therefore so is $ \check DG(1-\tilde P_\lambda)$. Manifestly, $\pi\circ \iota=1$. Moreover, by construction, the family $G$ satisfies $(\check D\tilde D+\tilde D\check D)G = \iota \tilde P_0- \text{id}$, so that 
\[
\tilde D (\check D G) (1-\tilde P_\lambda)+\check D G(1-\tilde P_\lambda)\tilde D = (\tilde P_0 -1)(1-\tilde P_\lambda)=\tilde P_0-1-\tilde P_0+\tilde P_\lambda=\tilde P_\lambda - 1.
\]
This verifies that we have constructed the deformation retraction in Equation \ref{eq: firstdefretract}. The deformation retraction of the lemma follows from the usual lift of a deformation retraction of complexes to their symmetric algebras, (see, for example, Proposition 2.5.5 of \cite{othesis}).  
\end{proof}

\begin{corA}
There is a deformation retraction
\[
\defretract{(\Gamma(\cdot, \Sym(\Pi(\tilde K_\lambda\oplus \tilde K_\lambda^\vee[1]))),\tilde D+\Delta)}{\shObq}{\iota'_\lambda}{\pi'_\lambda}{\eta'_\lambda}
\]
of sheaves of $\cinfty$-modules over $U_\lambda$, where $\iota'_{\lambda}$, $\pi'_\lambda$, and $\eta'_\lambda$ are induced from smooth bundle maps. Here, $\Delta$ is defined on $\Sym(\Pi(K_\lambda\oplus K_\lambda^\vee[1]))$ analogously to how it is defined on $\Obq$. 
\end{corA}
\begin{proof}
This follows from a straightforward application of the homological perturbation lemma \ref{lem: hpl}.
\end{proof}

As a result of Corollary A, we are now roughly in the situation of Section \ref{subsec: findimfamilies}; namely, we have reduced by Corollary A the computation of the cohomology of $\shObq$ to the computation of the cohomology of the observables arising from a finite-rank vector bundle $W$ as in that section. We can now prove the advertised theorem:

\begin{proof}[Proof of Theorem \ref{thm: main}]
Let us denote by $\shObq_{\lambda}$ the complex of sheaves on $B$
\[
(\Gamma(\cdot, \Sym(\Pi(\tilde K_\lambda\oplus \tilde K_\lambda^\vee[1]))),\tilde D+\Delta).
\]
Combining Corollary A and Lemma \ref{lem: famquantobs} (as well as the discussion immediately after Lemma \ref{lem: famquantobs}), we have the quasi-isomorphisms 
\[
\begin{tikzcd}
\Gamma(\cdot, \Det D)\cong\Gamma(\cdot, \Lambda^{top}\tilde K_\lambda)\ar[r,shift left = 0.5 ex,"\iota''_\lambda"] &\shObq_{\lambda}\ar[l, shift left = 0.5 ex, "\pi''_\lambda"] \ar[r,shift left = 0.5 ex,"\iota'_\lambda"]& \ar[l,shift left = 0.5 ex, "\pi'_\lambda"]\shObq
\end{tikzcd}
\]
of sheaves of $\cinfty$-modules over $U_\lambda$. We need to verify that the compositions $\pi''_\lambda\circ \pi'_\lambda$ respect the transition functions on $U_\lambda \cap U_\mu$. Then, we can define $\Phi\mid_{U_\lambda}=\pi''_\lambda\circ \pi'_\lambda$ and the first statement of the theorem will be proved. Let us assume that $\mu>\lambda$.  We need to show that 
\[(-1)^{\rk K_{(\lambda,\mu)}^+(\rk K_{(\lambda,\mu)}^+-1)/2}(id\otimes\Pflm)\circ \pi''_{\lambda}\circ \pi'_\lambda=\pi''_{\mu}\circ  \pi'_\mu.\]

Note that we can write $\tilde K_\mu = \tilde K_\lambda\oplus\tilde  K_{(\lambda,\mu)}$, and this decomposition induces a projection $\tilde P_{\mu}^\lambda: \tilde K_\mu \to \tilde K_\lambda$ and an inclusion $K_\lambda\subset K_\mu$, which we extend to cochain maps $\pi_{\mu}^\lambda: \shObcl_{\mu}\to \shObcl_{\lambda}$, $\iota^\lambda_\mu:\shObcl_\lambda\to \shObcl_\mu$. In fact, we have a deformation retraction
\[
\defretract{\shObcl_{\lambda}}{\shObcl_{\mu}}{\iota_{\mu}^{\lambda}}{\pi_\mu^{\lambda}}{\eta_{\mu}^{\lambda}},
\] 
of sheaves of $\cinfty_{U_\lambda\cap U_\mu}$-modules, where $\eta^\lambda_\mu = \pi_\mu (\eta_\mu - \eta_\lambda)\iota_\mu$ . This retraction perturbs, by the homological perturbation lemma, to a deformation retraction
\[
\defretract{\shObq_{\lambda}}{\shObq_{\mu}}{\iota_{\mu}^{'\lambda}}{\pi_\mu^{'\lambda}}{\eta_{\mu}^{'\lambda}}.
\]
Consider the following three diagrams of complexes of sheaves of $\cinfty$-modules on $U_\lambda\cap U_\mu$:

\begin{equation}
\label{eq: diag1}
\begin{tikzcd}
\Gamma(\cdot, \Det D)\cong\Gamma(\cdot, \Lambda^{top}\tilde K_\lambda)\ar[d,"\pm id\otimes\Pflm"] &\shObq_{\lambda,0}\ar[l,"p_\lambda"'] \ar[d,"e^{\tilde D_{\mu}^\lambda} i_{\mu}^\lambda"]\\
\Gamma(\cdot, \Det D)\cong\Gamma(\cdot, \Lambda^{top}\tilde K_\mu)&\shObq_{\mu,0}\ar[l,"p_\mu"'] ,
\end{tikzcd}
\end{equation}

\begin{equation}
\label{eq: diag2}
\begin{tikzcd}
 \shObq_{\lambda,0} \ar[d,"e^{\tilde D_{\mu}^\lambda} i_{\mu}^\lambda"]&\shObq_{\lambda}\ar[l,"e^{\tilde D_\lambda }"'] \ar[d,"\iota_{\mu}^{'\lambda}"']\\
\shObq_{\mu,0} &\shObq_{\mu}\ar[l,"e^{\tilde D_\mu}"']
\end{tikzcd},
\end{equation}

\begin{equation}
\label{eq: diag3}
\begin{tikzcd}
 \shObq_{\lambda}& \ar[l,"\pi'_\lambda"']\ar[ld,"\pi'_\mu"]\shObq\\
 \shObq_{\mu}\ar[u,"\pi_{\mu}^{'\lambda}"']&
\end{tikzcd}.
\end{equation}
Here, the maps $p_\mu$ and $p_\lambda$ are simply projection onto the $\Lambda^{top}$ summand; $i^\lambda_\mu$ is induced from the inclusion $\tilde K_{\lambda}\subset \tilde K_{\mu}$; and $\tilde D_{\lambda}$ (resp. $\tilde D_\mu, \tilde D^\lambda_\mu$) is the degree zero observable in $\shObq_{\lambda,0}$ (resp. $\shObq_{\mu,0},\shObq_{\mu,0}$) constructed from the restriction of $\tilde D$ to $\tilde K_\lambda$ (resp. $\tilde K_\mu$, $\tilde K_{(\lambda,\mu)})$ in the same way as in the discussion immediately preceding Lemma \ref{lem: multbypartfunc}.

The commutativity of these diagrams is shown in Lemma \ref{lem: diagrams}. They combine into a (non-commutative!) diagram of complexes of sheaves of $\cinfty$-modules on $U_\lambda\cap U_\mu$

\begin{equation}
\label{eq: bigdiagram}
\begin{tikzcd}
\Gamma(\cdot, \Det D)\cong\Gamma(\cdot, \Lambda^{top}\tilde K_\lambda)\ar[d,"\pm id\otimes \Pflm"] &\shObq_{\lambda,0}\ar[l,"p_\lambda"'] \ar[d,"e^{\tilde D_{\mu}^\lambda} i_{\mu}^\lambda"]& \shObq_{\lambda}\ar[l,"e^{\tilde D_\lambda }"'] \ar[d,"\iota_{\mu}^{'\lambda}"',shift right = 0.5ex]& \ar[l,"\pi'_\lambda"']\ar[ld,"\pi'_\mu"]\shObq\\
\Gamma(\cdot, \Det D)\cong\Gamma(\cdot, \Lambda^{top}\tilde K_\mu)&\shObq_{\mu,0}\ar[l,"p_\mu"'] &\shObq_{\mu}\ar[l,"e^{\tilde D_\mu}"'] \ar[u,"\pi_{\mu}^{'\lambda}"',shift right = 0.5ex]& 
\end{tikzcd};
\end{equation}
the diagram is not commutative on the nose, since the composition $\iota_\mu^{'\lambda}\pi_\mu^{'\lambda}$ is only the identity up to the homotopy $\eta_\mu^{'\lambda}$. However, this weaker version of commutativity will be enough to show the desired fact, namely that 
\begin{equation}
\label{eq: desiredformula}
 (-1)^{\rk K_{(\lambda,\mu)}^+(\rk K_{(\lambda,\mu)}^+-1)/2}id\otimes\Pflm \circ p_\lambda \circ e^{\tilde D_\lambda}\circ\pi'_\lambda=p_\mu \circ e^{\tilde D_\mu}\circ\pi'_\mu.
\end{equation}
(Here, we are using the fact that $\pi''_\lambda=p_\lambda e^{\tilde D_\lambda}$.) Namely, consider the following chain of equalities:
\begin{align*}
\pm\Pflm \circ p_\lambda \circ e^{\tilde D_\lambda}\circ\pi'_\lambda&= p_\mu \circ e^{\tilde D^\lambda_\mu } i^\lambda_\mu \circ e^{\tilde D_\lambda}\circ\pi'_\lambda\\
&= p_\mu \circ e^{\tilde D_\mu } \circ \iota^{'\lambda}_\mu\circ\pi'_\lambda\\
&=p_\mu \circ e^{\tilde D_\mu } \circ \iota^{'\lambda}_\mu\circ\pi^{'\lambda}_\mu \circ \pi'_\mu\\
&= p_\mu \circ e^{\tilde D_\mu}\circ\pi'_\mu- p_\mu \circ e^{\tilde D_\mu }\circ [d,\eta^{'\lambda}_\mu]\circ \pi'_\mu\\
&=p_\mu \circ e^{\tilde D_\mu}\circ\pi'_\mu- d\circ p_\mu \circ e^{\tilde D_\mu }\circ \eta^{'\lambda}_\mu\circ \pi'_\mu-p_\mu \circ e^{\tilde D_\mu }\circ \eta^{'\lambda}_\mu\circ d\circ \pi'_\mu,
\end{align*}
where $d$ is our generic term for a differential on the relevant chain complex; the second term in the last expression is a map $\shObq \to \Gamma(\cdot, \Lambda^{top}\tilde K_\mu)$ with image in the space of exact cochains in the latter space, which is precisely zero. Moreover, $\eta^{'\lambda}_\mu$ has image in strictly negative degrees in $\shObq_\mu$ and hence $p_\mu \circ e^{\tilde D_\mu }\circ \eta^{'\lambda}_\mu=0$. This proves Equation \ref{eq: desiredformula}, and hence the first statement of the theorem.

It remains to check only the claim about the determinant section. The determinant of $D$ is defined to be, in the trivializing neighborhood $U_\lambda$, $\det D^+_\lambda$. The constant function $1\in \cinfty(B)$ is a closed observable in $\shObq(B)$, and by the proof of Lemma \ref{lem: pfaffiso}, $\pi^{''}_\lambda\circ\pi'_\lambda(1)=\pf(\tilde D_\lambda)=(-1)^{\rk K_\lambda^+(\rk K_\lambda^+-1)/2}\det(D^+_\lambda)$ (up to, perhaps, a sign).   
\end{proof}

\begin{remark}
\label{rem: iotafails}
With a little bit of work, one can show that the maps $\iota'_\lambda\circ \iota''_\lambda$ do not necessarily respect the transition functions of $\Det D$; instead, we claim but do not prove that one can specify an explicit homotopy relating $\iota'_\lambda\circ \iota''_\lambda$ to $\iota'_\mu\circ \iota''_\mu$ on $U_\lambda\cap U_\mu$.
\end{remark}

\begin{lemma}
\label{lem: diagrams}
The diagrams in Equations \ref{eq: diag1}, \ref{eq: diag2}, and \ref{eq: diag3} commute.
\end{lemma}
\begin{proof}
For the diagram of Equation \ref{eq: diag1}, the lemma is a consequence of the commutativity of the squares
\begin{equation}
\label{eq: twodiagrams}
\begin{tikzcd}
\cinfty_B\mid_{U_\lambda\cap U_\mu}\ar[d,"\pm\Pflm\cdot"]&\ar[l,"id"'] \cinfty_B\mid_{U_\lambda\cap U_\mu}\ar[d,"e^{\tilde D^\lambda_\mu}\cdot"]\\
\Gamma(\cdot, \Lambda^{top}\tilde K_{(\lambda,\mu)})& \ar[l,"p_{(\mu,\lambda)}"']\shObq_{(\lambda,\mu),0}
\end{tikzcd},
\begin{tikzcd}
\Gamma(\cdot,\Lambda^{top}\tilde K_\lambda) \ar[d,"id"]&\ar[l,"p_\lambda"']\ar[d,"id"] \shObq_{\lambda,0}\\
\Gamma(\cdot, \Lambda^{top}\tilde K_\lambda) & \ar[l,"p_{\lambda}"']\shObq_{\lambda,0}
\end{tikzcd},
\end{equation}
where $\shObq_{(\lambda, \mu),0}$ is defined analogously to $\shObq_{\lambda,0}$ by using $\tilde K_{(\lambda,\mu)}$ instead of $\tilde K_\lambda$. The diagrams are diagrams of sheaves of $\cinfty_{U_\lambda\cap U_\mu}$-modules. The extra sign in front of $\Pflm$ accounts for the discrepancy between $\pf( \tilde D^\mu_\lambda)$ and $\det (D^\mu_\lambda)^+$ (see Remark \ref{rem: pfafftodet}). 

 Upon taking the tensor product over $\cinfty_{U_\lambda\cap U_\mu}$ of the two diagrams in equation \ref{eq: twodiagrams}, and using the evident isomorphisms over $U_\lambda\cap U_\mu$
\begin{align*}
\Gamma(\cdot, \Lambda^{top}\tilde K_\mu)&\cong \Gamma(\cdot, \Lambda^{top}\tilde K_\lambda \otimes \Lambda^{top}\tilde K_{(\lambda,\mu)})\\
\shObq_{\mu,0} &\cong \shObq_{\lambda,0}\otimes_{\cinfty_{U_\lambda\cap U_\mu}} \shObq_{(\lambda,\mu),0},
\end{align*}
one finds the commutativity of the diagram of Equation \ref{eq: diag1}.

The commutativity of the diagram of Equation \ref{eq: diag2} follows from the fact that $\iota_\mu^{'\lambda}=\iota_\mu^\lambda=i^\lambda_\mu$, which in turn follows from the fact that $\eta^\lambda_\mu \iota^\lambda_\mu=0$.

The commutativity of the diagram of Equation \ref{eq: diag3} is a consequence of Lemma \ref{lem: compofprojections}. We have constructed $\eta_\mu^\lambda$ so that $\eta_\lambda$ is the composite contracting homotopy of $\eta_\mu$ and $\eta^\lambda_\mu$, and one can directly check relations \ref{eq: cond1}, \ref{eq: cond2}, \ref{eq: cond3}.
\end{proof}

\appendix
\section{The homological perturbation lemma}
\label{sec: app}
In this appendix, we list the homological algebraic lemmas that we use for the proof of Theorem \ref{thm: main}.

\begin{definition}
Let $(V,d_V)$ and $(W,d_W)$ be two cochain complexes. \textbf{A deformation retraction of $V$ onto $W$}, denoted 
\[
\begin{tikzcd}
(W,d_W)\arrow[r,shift left = .5 ex," \iota"] & (V,d_V)\arrow[l, shift left = .5 ex," \pi"]\arrow[loop right,distance = 4em,start anchor = {[yshift = 1ex]east},end anchor = {[yshift=-1ex]east}]{}{\eta}
\end{tikzcd},
\]
is a triple of maps $(\iota, \pi, \eta)$, where $\iota$ and $\pi$ are cochain maps $W\to V$ and $V\to W$, such that $\pi\iota=id_W$ and $\eta$ is a cochain homotopy between $\iota\pi$ and $id_V$.
\end{definition}
The following is the main theorem discussed in \cite{crainic}.
\begin{lemma}[Homological Perturbation Lemma]
\label{lem: hpl}
Suppose that 

\[
\begin{tikzcd}
(W,d_W)\arrow[r,shift left = .5 ex," \iota"] & (V,d_V)\arrow[l, shift left = .5 ex," \pi"]\arrow[loop right,distance = 4em,start anchor = {[yshift = 1ex]east},end anchor = {[yshift=-1ex]east}]{}{\eta}
\end{tikzcd}
\]
is a deformation retraction, and suppose that $\delta_V$ is a degree 1 operator such that $d_V+\delta_V$ is a differential and $(1-\delta_V\eta)$ is invertible. Then 

\[
\begin{tikzcd}
(W,d_W+\delta_W)\arrow[r,shift left = .5 ex," \iota'"] & (V,d_V+\delta_V)\arrow[l, shift left = .5 ex," \pi'"]\arrow[loop right, distance = 4em,start anchor = {[yshift = 1ex]east},end anchor = {[yshift=-1ex]east}]{}{\eta'}
\end{tikzcd}
\]
is a deformation retraction, where

\begin{align*}
\delta_W &= \pi (1-\delta_V \eta)^{-1}\delta_V \iota\\
\iota' & = \iota + \eta (1-\delta_V\eta)^{-1}\delta_V\iota\\
\pi' & = \pi + \pi (1-\delta_V\eta)^{-1}\delta_V\eta\\
\eta' & = \eta+ \eta(1-\delta_V\eta)^{-1}\delta_V\eta
\end{align*}
are the perturbed data of the deformation retraction.
\end{lemma}

\begin{lemma}
\label{lem: comphe}
Given two deformation retractions
\[
\begin{tikzcd}
(V_1,d_1)\arrow[r,shift left = .5 ex," \iota_1"] & (V_2,d_2)\arrow[l, shift left = .5 ex," \pi_1"]\arrow[loop right, distance = 4em,start anchor = {[yshift = 1ex]east},end anchor = {[yshift=-1ex]east}]{}{\eta_1}
\end{tikzcd}
\]

and

\[
\begin{tikzcd}
(V_2,d_2)\arrow[r,shift left = .5 ex," \iota_2"] & (V_3,d_3)\arrow[l, shift left = .5 ex," \pi_2"]\arrow[loop right, distance = 4em,start anchor = {[yshift = 1ex]east},end anchor = {[yshift=-1ex]east}]{}{\eta_2}
\end{tikzcd},
\]

the composite 

\[
\begin{tikzcd}
(V_1,d_1)\arrow[r,shift left = .5 ex," \iota_2 \iota_1"] & (V_3,d_3)\arrow[l, shift left = .5 ex," \pi_1 \pi_2"]\arrow[loop right, distance = 4em,start anchor = {[yshift = 1ex]east},end anchor = {[yshift=-1ex]east}]{}{\eta_2+\iota_2\eta_1\pi_2}
\end{tikzcd}
\]

is also a deformation retraction.
\end{lemma}  
\begin{proof}
Direct computation. 
\end{proof}

\begin{lemma}
\label{lem: compofprojections}
Suppose given two deformation retractions
\[
\defretract{(V_1,d_1)}{(V_2,d_2)}{\iota_1}{\pi_1}{\eta_1},
\]
\[
\defretract{(V_2,d_2)}{(V_3,d_3)}{\iota_2}{\pi_2}{\eta_2},
\]
and let $\Delta_2$ and $\Delta_3$ be two ``small'' perturbations (i.e. perturbations satisfying the hypotheses of the homological perturbation lemma). Suppose further that 
\begin{align}
\label{eq: cond1}
\eta_2\Delta_3\iota_2&=0\\
\label{eq: cond2}
\eta_1\Delta_2\iota_1&=0\\
\label{eq: cond3}
\Delta_2 =\pi_2&\Delta_3 \iota_2.
\end{align}
Finally, let 
\begin{align*}
\pi_3&:= \pi_1\pi_2\\
\iota_3&:=\iota_2\iota_1\\
\eta_3&:=\eta_2+\iota_2\eta_1\pi_2
\end{align*}
be the data of the composite retraction 
\[
\defretract{(V_1,d_1)}{(V_3,d_3)}{\iota_3}{\pi_3}{\eta_3},
\]
and $\delta_{3\to 1}$ (resp. $\delta_{2\to 1}$) the differential induced by $\Delta_3$ (resp. $\Delta_2$) on $V_1$ using the retraction $(\iota_1,\pi_1,\eta_1)$ (resp. $(\iota_3,\pi_3,\eta_3)$). Then, $\delta_2=\Delta_2$ (where $\delta_2$ is the differential on $V_2$ induced from $\Delta_3$ via the retraction of $V_3$ onto $V_2$), $\delta_{3\to 1}=\delta_{2\to 1}$, and $\pi'_3=\pi'_1\pi'_2$. 
\end{lemma}
\begin{remark}
The above lemma says that, for the purposes of computing the perturbed differential on $V_1$ and the perturbed projection $V_3\to V_1$, one can either perturb the data to $V_2$ and then perturb to $V_1$, or perturb directly to $V_1$ and the results will be the same. 
\end{remark}
\begin{proof}
\[
\delta_2 = \pi_2 \sum_{n=0}^\infty \left(\Delta_3\eta_2\right)^n\Delta_3\iota_2 = \pi_2\Delta_3\iota_2=\Delta_2,
\]
where the first equality is the formula given by the homological perturbation lemma, the second equality is given by assumption \ref{eq: cond1}, and the last equality is given by assumption \ref{eq: cond3}. Similarly, 
\[
\delta_{2\to 1} =\pi_1 \sum_{n=0}^\infty \left(\Delta_2\eta_1\right)^n\Delta_2\iota_1 =\pi_1\Delta_2\iota_1,
\]
and 
\[
\delta_{3\to 1} = \pi_1\pi_2\sum_{n=0}^\infty \left( \Delta_3 (\eta_2+\iota_2\eta_1\pi_2)\right)^n \Delta_3 \iota_2\iota_1=\pi_1\pi_2\Delta_3\iota_2\iota_1=\pi_1\Delta_2\iota_1,
\]
which proves the second stated claim. Now, 
\begin{align*}
\pi_3'&=\pi_1\pi_2\sum_{n=0}^\infty (\Delta_3(\eta_2+\iota_2\eta_1\pi_2))^n=\pi_1\pi_2\sum_{n,m=0}^\infty (\Delta_3\iota_2\eta_1\pi_2)^n(\Delta_3\eta_2)^m\\
&=\pi_1\sum_{n=0}^\infty (\pi_2\Delta_3\iota_2\eta_1)^n \pi_2\sum_{m=0}^\infty (\Delta_3\eta_2)^m=\pi_1\sum_{n=0}^\infty (\Delta_2\eta_1)^n \pi_2\sum_{m=0}^\infty (\Delta_2\eta_2)^m=\pi_1'\pi_2',
\end{align*}
which is the last stated claim.
\end{proof}
\bibliography{references}
\bibliographystyle{alpha}

\end{document}